\documentclass{llncs}
\usepackage{amssymb, amsmath}

\usepackage{verbatim}
\usepackage{graphicx}
\usepackage{subfigure}

\usepackage{color}
\usepackage[normalem]{ulem}

\newcommand{\oldversion}[1]{{}}

\newcommand{\remove}[1]{}

\newcommand{\NP}{\mbox{\bfseries NP}}

\newenvironment{proofof}[1]{\par\addvspace\topsep\noindent
{\bf Proof #1:} \ignorespaces }{\qed}

\usepackage[pagebackref=true,breaklinks=true]{hyperref}

\begin{document}

\title{Drawing Partially Embedded and\\ Simultaneously Planar Graphs\thanks{A preliminary version of this paper appeared in~\cite{cfglms-dpespg-14}.}}

\author{ Timothy~M.~Chan\inst{1}
\and Fabrizio Frati\inst{2}
    \and Carsten~Gutwenger\inst{3}
    \and Anna~Lubiw\inst{1}
    \and Petra~Mutzel\inst{3}
    \and Marcus~Schaefer\inst{4}
}

\institute{Cheriton School of Computer Science, University of Waterloo, Canada.\\
\email{\{tmchan, alubiw\}@uwaterloo.ca}
\and
School of Information Technologies, The University of Sydney, Australia.\\
\email{fabrizio.frati@sydney.adu.au}
\and
Technische Universit\"{a}t Dortmund, Dortmund, Germany.\\
\email{\{carsten.gutwenger, petra.mutzel\}@tu-dortmund.de}
\and DePaul University, Chicago, Illinois, USA.\\
\email{mschaefer@cdm.depaul.edu}
}


\maketitle

\begin{abstract}
We investigate the problem of constructing planar drawings with few bends for two related problems, the \emph{partially embedded graph} problem---to extend a straight-line planar drawing of a subgraph to a planar drawing of the whole graph---and the \emph{simultaneous planarity} problem---to find planar drawings of two graphs that coincide on shared vertices and edges. In both cases we show that if the required planar drawings exist, then there are planar drawings with a linear number of bends per edge and, in the case of simultaneous planarity, a constant number of crossings between every pair of edges.  Our proofs provide efficient algorithms if the combinatorial embedding of the drawing is given. Our result on partially embedded graph drawing  generalizes a classic result by Pach and Wenger which shows that any planar graph can be drawn with a linear number of bends per edge if the location of each vertex is fixed.
 \end{abstract}

\section{Introduction}
\label{sec:intro}

In many practical applications we wish to draw a planar graph while satisfying some geometric or topological constraints. One natural situation is that we have a drawing of part of the graph and wish to extend it to a planar drawing of the whole graph. Pach and Wenger~\cite{PW01} considered a special case of this problem.  They showed that any planar graph can be drawn with its vertices lying at pre-assigned points in the plane and with a linear number of bends per edge. In this case the pre-drawn subgraph has no edges.

If the pre-drawn subgraph $H$ has edges, a planar drawing of the whole graph $G$ extending the given drawing $\cal H$ of $H$ may not exist. Angelini et al.~\cite{ABFJKPR10} gave a linear-time algorithm for the corresponding decision problem; the algorithm returns, for a positive answer, a planar embedding of $G$ that {\em extends} that of $\cal H$ (i.e., if we restrict the embedding of $G$ to the edges and vertices of $H$, we obtain the embedding corresponding to $\cal H$). If one does not care about maintaining the actual planar drawing of $H$ this is the end of the story, since standard methods can be used to find a straight-line planar drawing of $G$ in which the drawing of $H$ is topologically equivalent to the one of $\cal H$. In this paper we show how to draw $G$ while preserving the actual drawing $\cal H$ of $H$, so that each edge has a linear number of bends. This bound is worst-case optimal, as proved by Pach and Wenger~\cite{PW01} in the special case in which $H$ has no edges.

A result analogous to ours was claimed by Fowler et al.~\cite{fjks-crp-11} for the special case in which $H$ has the same vertex set as $G$. Their algorithm draws the edges of $G$ one by one, in any order so that edges connecting distinct connected components of $H$ precede edges within the same connected component of $H$; each edge is drawn as a curve with the minimum number of bends. Fowler et al. claim that their algorithm constructs drawings with a linear number of bends per edge. However, we prove that there exists a tree, a planar drawing of its vertex set, and an order of the edges of the tree, such that drawing the edges in the given order as curves with the minimum number of bends results in some edges having an exponential number of bends.

The second graph drawing problem we consider is the \emph{simultaneous planarity} problem~\cite{BKR?}, also known as ``simultaneous embedding with fixed edges (SEFE)''.  The SEFE problem is strongly related to the partially embedded graph problem and---in a sense we will make precise later---generalizes it. We are given two planar graphs $G_1$ and $G_2$ that share a \emph{common subgraph} $G$ (i.e., $G$ is composed of those vertices and edges that belong to both $G_1$ and $G_2$). We wish to find a \emph{simultaneously planar drawing}, i.e., a planar drawing of $G_1$ and a planar drawing of $G_2$ that coincide on $G$. Graphs $G_1$ and $G_2$ are \emph{simultaneously planar} if they admit such a drawing. Both $G_1$ and $G_2$ may have {\em private} edges that are not part of $G$. In a simultaneous planar drawing the private edges of $G_1$ may cross the private edges of $G_2$. The simultaneous planarity problem arises in information visualization when we wish to display two relationships on two overlapping element sets.

The decision version of the simultaneous planarity problem is not known to be \NP-complete, or to be solvable in polynomial time, though it is known to be \NP-complete if more than two graphs are given~\cite{GJPSS06}. However, there is a combinatorial characterization of simultaneous planarity, based on the concept of a ``compatible embedding'', due to J{\"u}nger and Schulz~\cite{JS} (see below for details). Erten and Kobourov~\cite{EK}, who first introduced the problem, gave an efficient drawing algorithm for the special case where the two graphs share vertices but no edges.  In this case, a simultaneous planar drawing on a polynomial-size grid always exists in which each edge has at most two bends and therefore any two edges cross at most nine times, see~\cite{dl-seogpc-07,EK,k-setbepa-06}. In this paper we show that if two graphs have a simultaneous planar drawing, then there is a drawing on a polynomial-size grid in which every edge has a linear number of bends and in which any two edges cross at most 24 times. Our result is algorithmic, assuming a compatible embedding is given.

\subsection{Realizability Results}

Our paper addresses the following two drawing problems:

\begin{description}

\item[{\bf Planarity of a partially embedded graph (PEG).}]
Given a planar graph $G$ and a straight-line planar drawing ${\cal H}$ of a subgraph $H$ of $G$,
find a planar drawing of $G$ that extends ${\cal H}$ (see~\cite{ABFJKPR10,JKR13}).

\item[{\bf Simultaneous planarity (SEFE).}]
Given two planar graphs $G_1$ and $G_2$ that share a subgraph $G$,
find a simultaneous planar drawing of $G_1$ and $G_2$ (see~\cite{BKR?}).
\end{description}

We prove the following results:

\begin{theorem}[Realizing a Partially Embedded Graph]\label{thm:PWG}
Let $G$ be an $n$-vertex planar graph, let $H$ be a subgraph of $G$, and let $\cal H$ be a straight-line planar drawing of $H$. Suppose that $G$ has a planar embedding $\cal E$ that extends the one of $\cal H$.
Then we can construct a planar drawing of $G$ in $O(n^2)$-time which realizes $\cal E$, extends $\cal H$, and has at most $72 |V(H)|$ bends per edge.
 \end{theorem}

Theorem~\ref{thm:PWG} generalizes Pach and Wenger's classic result, which corresponds to the special case in which the pre-drawn subgraph has no edges.

\begin{theorem}[Realizing a Simultaneous Planar Embedding]
\label{thm:sim-draw}
Let $G_1$ and $G_2$ be simultaneously planar graphs on a total of $n$ vertices with a shared subgraph $G$.
If we are given a compatible embedding of the two graphs, then we can construct in $O(n^2)$ time a drawing that realizes the compatible embedding, and in which
any private edge of $G_1$ and any private edge of $G_2$
intersect at most $24$ times. In addition, we can ensure either one of the following two properties:

\begin{enumerate}
\item[$(i)$] each edge of $G$ is straight, and each private edge of $G_1$ and of $G_2$ has at most $72n$ bends; also, vertices, bends, and crossings lie on an $O(n^2) \times O(n^2)$ grid; or
\item[$(ii)$] each edge of $G_1$ is straight and each private edge of $G_2$ has at most $72 |V(G_1)|$ bends per edge.
\end{enumerate}
\end{theorem}

 Theorem~\ref{thm:PWG} provides a weak form of Theorem~\ref{thm:sim-draw}: If $G_1$ and $G_2$ are simultaneously planar, they admit a compatible embedding. Take any straight-line planar drawing of $G_1$ realizing that embedding and extend the induced drawing of $G$ to a drawing of $G_2$. By Theorem~\ref{thm:PWG}, we obtain a simultaneous planar drawing where each edge of $G_1$ is straight and each private edge of $G_2$ has at most $72|V(G_1)|$ bends per edge. Our stronger result of 24 crossings between any two edges is  obtained by modifying the proof  of Theorem~\ref{thm:PWG}, rather than applying that result directly.

Grilli et al.~\cite{Grilli2014} very recently and independently proved a result in some respect stronger than Theorem~\ref{thm:sim-draw}. They showed that two simultaneously planar graphs have a simultaneous planar drawing with at most $9$ bends per edge, vastly better than our $72n$ bound. On the other hand, our bound of $24$ crossings per pair of edges is better than the bound of $100$ that can be derived from their result. Also, our algorithm allows us to construct simultaneous planar drawings in which each edge of one graph is straight or in which vertices, bends, and crossings lie on a polynomial-size grid. The former feature is not achievable by means of Grilli et al.'s algorithm; the latter one could be obtained from Grilli et al.'s result, at the expense of increasing the number of bends per edge to $300n$ (which corresponds to the number of crossings on a single private edge).

\subsection{Related Work}

The decision version of simultaneous planarity generalizes partially embedded planarity: given an instance $(G,H,\cal H)$ of the latter problem, we can augment $\cal H$ to a drawing of a $3$-connected graph $G_1$ and let $G_2 = G$.  Then $G_1$ and $G_2$ are simultaneously planar if and only if $G$ has a planar embedding extending $\cal H$. In the other direction, the algorithm~\cite{ABFJKPR10} for testing planarity of partially embedded graphs solves the special case of the simultaneous planarity problem in which the embedding of the common graph $G$ is fixed (which happens, e.g., if $G$ or one of the two graphs is $3$-connected).

Several optimization versions of partially embedded planarity and simultaneous planarity are \NP-hard. Patrignani showed that testing whether there is a straight-line drawing of a planar graph $G$ extending a given drawing of a subgraph of $G$ is \NP-complete~\cite{P06}, so bend minimization in partial embedding extensions is \NP-complete; Patrignani's result holds  even if a combinatorial embedding of $G$ is given.\footnote{Patrignani does not explicitly claim \NP-completeness in the case in which the embedding of $G$ is fixed, but that can be concluded by checking his construction; only the variable gadget, pictured in his Figure~3, needs minor adjustments.}
Bend minimization in simultaneous planar drawings is \NP-hard, since it is \NP-hard to decide whether there is a straight-line simultaneous drawing~\cite{EGJPSS08}.
Crossing minimization in simultaneous planar drawings is also \NP-hard, as follows from an \NP-hardness result on \emph{anchored planar drawings} by Cabello and Mohar~\cite{Cabello-Mohar}; see Theorem~\ref{thm:SPC} in Section~\ref{sec:S2} for a slightly stronger result.

Di Giacomo et al.~\cite{ddlmw-psetgpd-09} studied the special case of PEG in which the $n$-vertex graph $G$ to be drawn is a tree. They showed that, given a drawing $\cal H$ of a subtree $H$ of $G$, a drawing of $G$ extending $\cal H$ can be computed in $O(n^2\log n)$ time so that each edge of $G$ has at most $1+2\lceil{|V(H)|/2}\rceil$ bends.

Further, as mentioned above, the special cases of PEG and SEFE in which there are no edges in the pre-drawn subgraph and in the common subgraph have been already studied.

Concerning PEG, Pach and Wenger~\cite{PW01} proved the following result: given an $n$-vertex planar graph $G$ with fixed vertex locations, a planar drawing of $G$ in which each edge has at most $120n$ bends can be constructed in $O(n^2)$ time. They also proved that such a bound is asymptotically tight in the worst case.
Regarding the constant, Badent et al.~\cite{bgl-dcgcps-08} improved the bound to $3n+2$ bends per edge.
Biedl and Floderus~\cite{Biedl-Floderus} considered the more general problem of drawing an $n$-vertex planar graph on fixed vertex locations where the drawing is constrained to lie inside a $k$-vertex polygon.   They show that there is a drawing with $O(n+k)$ bends per edge.

Concerning SEFE, Di Giacomo and Liotta~\cite{dl-seogpc-07} and independently Kammer~\cite{k-setbepa-06} proved the following result: given two planar graphs $G_1$ and $G_2$ sharing some vertices and no edge with a total number of $n$ vertices, there exists an $O(n)$-time algorithm to construct a simultaneous planar drawing of $G_1$ and $G_2$ on a grid of size $O(n^2) \times O(n^2)$, where each edge has at most $2$ bends, hence there are at most $9$ crossings between any edge of $G_1$ and any edge of $G_2$. This improves upon a previous result of Erten and Kobourov~\cite{EK}. The algorithms in~\cite{dl-seogpc-07,EK,k-setbepa-06} make use of a drawing technique introduced by Kaufmann and Wiese~\cite{KW}.

Haeupler et al.~\cite{HJL} showed that if two simultaneously planar graphs $G_1$ and $G_2$ share a subgraph $G$ that is connected, then there is a simultaneous planar drawing in which any edge of $G_1$ and any edge of $G_2$ intersect at most once. Introducing vertices at crossing points yields a planar graph, and a straight-line drawing of that graph provides a simultaneous planar drawing with $O(n)$ bends per edge, $O(n)$ crossings per edge, and with vertices, bends, and crossings on an $O(n^2) \times O(n^2)$ grid. Our result generalizes this to the case where the common graph $G$ is not necessarily connected.

\subsection{Graph Drawing Terminology}\label{sec:definitions}

A \emph{rotation system} for a graph is a cyclic ordering of the edges incident to each vertex.
A rotation system of a connected graph determines its {\em facial walks}---the closed walks in which each edge $(u,v)$ is followed by the next edge $(v,w)$ in the cyclic order at $v$.
The {\em size} $|W|$ of a facial walk $W$ is the number of vertices on $W$, where we count vertex repetitions.  (Note that a graph that consists of a single vertex has a single facial walk of size 1; for any other connected graph the size of a facial walk is equal to the number of edges in the facial walk, counting repetitions.)
A rotation system is \emph{planar} if Euler's formula holds, i.e., $n - m + f = 2$ where $n$ is the number of vertices, $m$ is the number of edges, and $f$ is the number of facial walks.
A \emph{planar embedding} of a graph consists of a planar rotation system together with a specified outer face.  A fundamental result about connected planar graphs is that every planar drawing corresponds to a planar embedding, and conversely, every planar embedding can be realized as a planar drawing (and, in fact, as a straight-line planar drawing by F\'ary's theorem).
Furthermore, facial walks correspond to faces in the drawing.


These definitions do not handle the combinatorics of a planar drawing of a disconnected graph---namely the definition of planar embedding as stated above does not tell us how connected components nest into each other.


Following  J{\"u}nger and Schulz~\cite{JS}, we
define a \emph{topological embedding} of a (possibly non-connected) graph as follows: We specify a planar embedding for each connected component.  This determines a set of inner faces. For each connected component we specify a ``containing'' face, which may be an inner face of some other component or the unique outer face.   Furthermore, we forbid cycles of containment---in other words, if a connected component is contained in an inner face, which is contained in a component, etc., then this chain of containments must lead eventually to the unique outer face.

A {\em facial boundary} in a topological embedding of a graph is the collection of facial walks along the (not necessarily connected) boundary of a face. Each face (unless it is the outer face) has a distinguished facial walk we call the {\em outer} facial walk separating the remaining {\em inner} facial walks from the outer face of the embedding. The {\em size} of a facial boundary is the sum of the sizes of its facial walks.

A \emph{compatible embedding} of two planar graphs $G_1$ and $G_2$ consists of topological embeddings of $G_1$ and $G_2$ such that the common subgraph $G$ inherits the same topological embedding from $G_1$ as from $G_2$ (where a subgraph inherits a topological embedding in a straightforward way; in particular, if we remove an edge that disconnects the graph, the face containment is determined by the edge that was removed). J{\"u}nger and Schulz~\cite{JS} proved that $G_1$ and $G_2$ are simultaneously planar if and only if they have a compatible embedding. For that proof, they construct a simultaneous planar drawing of $G_1$ and $G_2$ by extending a drawing of $G$ (thus proving a form of our Theorem~\ref{thm:PWG}). However, their method does not yield any bounds on the number of bends or crossings.

\section{Partially Embedded Graphs}\label{sec:PEG}

In this section we prove Theorem~\ref{thm:PWG}; that is, we show how to construct a planar drawing of $G$ that extends the planar straight-line drawing ${\cal H}$ and has a linear number of bends per edge assuming that we are given a planar embedding of $G$ extending $\cal H$. It is sufficient to prove the result for a single face $F$ of $\cal H$, since the embedding of $G$ is given, and we know for each vertex and edge of $G$ which face of ${\cal H}$ it lies in, so the drawings in different faces of $\cal H$ do not interfere with each other.

Pach and Wenger~\cite{PW01} proved their upper bound on the number of bends needed to draw a graph with fixed vertex locations by drawing a tree with its leaves at the fixed vertex locations, and ``routing'' all the edges close to the tree, sometimes crossing the tree but never crossing each other. We want to use their approach, but we have to deal with a more general problem. Instead of fixed vertex locations we have fixed facial boundaries. The solution is natural: We contract each facial walk $W_i$ of $F$ to a single vertex $v_i$, fix a position for vertex $v_i$ inside $F$ near $W_i$, and then apply the Pach-Wenger method to draw the contracted graph on the fixed vertex locations $v_i$. We ensure that the contracted graph is drawn inside $F$, indeed we
stay a small distance away
from the boundary of $F$, inside a polygonal region $F'$ that is an ``inner approximation'' of $F$. Inside $F'$ we draw a tree $T$ with its leaves $v_i$ at the fixed vertex locations, while suitably bounding the size of $T$ so as to get our bound on the number of bends. We then route the edges of the contracted graph close to $T$ as  Pach and Wenger do. Finally, to retrieve the original, uncontracted graph,  we route the edges incident to $v_i$ to their true endpoint on the facial boundary $W_i$---these routes use the empty buffer zone between $F$ and $F'$.

We fill in the details of this argument in Section~\ref{sec:PPWG}, but before doing so we introduce ``inner approximations'' in Section~\ref{sec:AF},
and formalize the tree argument in Section~\ref{sec:EPP}.

To simplify notation, we use $n_A$ and $m_A$ for the number of vertices and edges in a graph (or subgraph) $A$.

\subsection{Approximating Faces}\label{sec:AF}

In the drawing $\cal H$, the face $F$ is a region of the plane
homeomorphic to a disc with holes.  Each facial walk of $F$ appears in the drawing as a  {\em closed polygonal arc}, i.e.~a sequence of straight-line segments joined in a path that returns to its starting point (repeated segments/vertices may occur).
We will refer to a facial walk and its drawing interchangeably.

We will approximate $F$ by offsetting each of its facial walks into the interior of $F$.
See Figure~\ref{fig:approx}.
Let $W_1$ be the outer facial walk of $F$, and let $W_2, \ldots, W_b$ be the inner facial walks.
An {\em inner $\varepsilon$-approximation of $W_i$} is a simple polygon $P_i$ (a closed polygonal arc with no self-intersections) such that:
\begin{enumerate}
\item $P_i$  is $\varepsilon${\em-close} to  $W_i$, meaning that every point of $P_i$ is within distance $\epsilon$ of a point of $W_i$,
\item the inner facial walk $W_i$ lies in the interior of $P_i$, for each $2\leq i\leq b$, and
\item the outer facial walk $W_1$ lies in the exterior of $P_1$.
\end{enumerate}
If in addition the $P_i$'s form a {\em polygonal region} (a simple polygon with holes) with $P_1$ as the outer polygon, then we say that the polygonal region is an  {\em inner $\varepsilon$-approximation of $F$}.
\remove{The {\em Hausdorff distance} $d_H(A,B)$ of two sets (in a space with metric $d$) is defined as\footnote{The underlying metric $d$ can be Euclidean or some other appropriate metric.}\\  {\begin{center} $\max\left\{\sup_{a \in A} \inf_{b \in B} d(a,b), \sup_{b \in B} \inf_{a \in A} d(a,b)\right\}$.\end{center}} Intuitively, the Hausdorff distance measures how far a point in one set can be from the other set. Sets $A$ and $B$ are {\em $\varepsilon$-close} if $d_H(A,B) < \varepsilon$. Then $A$ is an {\em inner $\varepsilon$-approximation of $B$} if they are $\varepsilon$-close and there is a $\delta>0$ so that all the points $\delta$-close to $A$ are a subset of $B$.
}
The next lemma shows that we can build inner $\varepsilon$-approximations of $F$. 

\begin{lemma}\label{cor:fw}
For any $\varepsilon>0$ we can efficiently construct an inner $\varepsilon$-approximation $F'$ of $F$.
\end{lemma}

\begin{figure}[tb]
\centering
\includegraphics[width=0.8\textwidth]{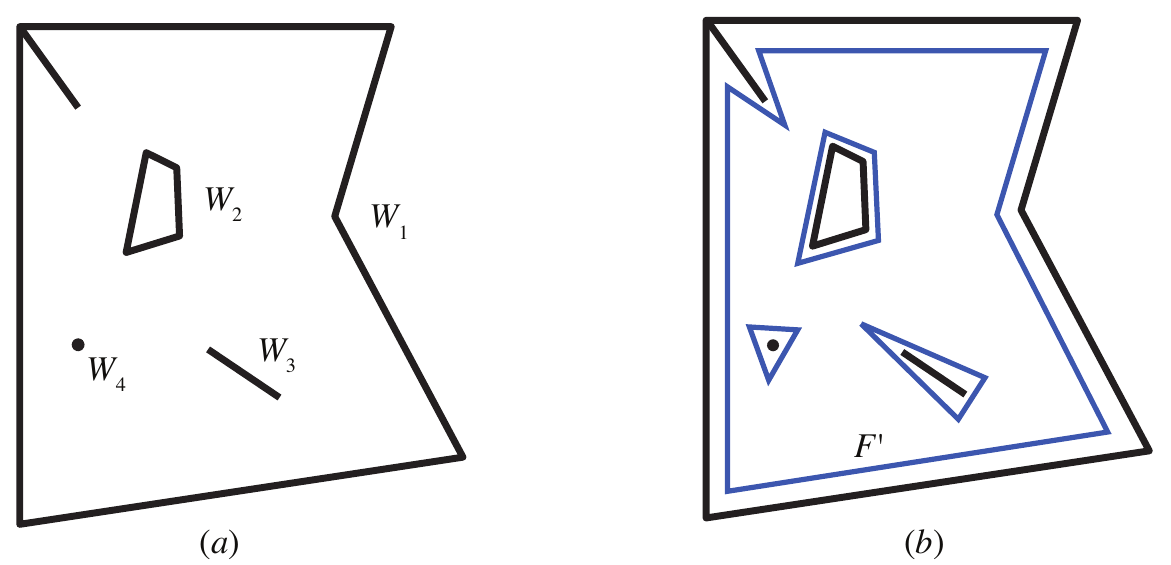}
\caption{(\emph{a}) A face $F$ with outer facial walk $W_1$ and inner facial walks $W_2, W_3, W_4$. (\emph{b}) An inner approximation $F'$ (heavy blue lines).}
\label{fig:approx}
\end{figure}

See Figure~\ref{fig:approx} for an illustration of Lemma~\ref{cor:fw}. To prove the lemma, we construct---for every sufficiently small $\varepsilon > 0$ and for every facial walk of $F$---an
inner ${\varepsilon}$-approximating polygon
$P_{\varepsilon}$ which
does not have too many bends, and so that the $P_{\varepsilon}$ are
{\em nested}
in the following sense: if $0 < \varepsilon' < \varepsilon$, then
$P_{\varepsilon'}$ lies in the interior of
$P_{\varepsilon}$ if $F$ is an inner face, and vice versa otherwise.
There are various ways to achieve this. Pach and Wenger~\cite{PW01} use the Minkowski sum of the facial walk (in their case the facial walk of a tree) and a square diamond centered at $0$.  We use a slightly different construction, because it seems easier (both computationally and conceptually) and it gives a slightly better bound on the number of bends (which is what we are most interested in): for the facial walk of an $n$-vertex tree, Pach and Wenger construct a
polygon with $4n-2$ vertices, while ours  
have $2n-2$ vertices. Our construction does have one disadvantage: the resulting drawings are tight, placing elements close together, for sharp (acute or obtuse) angles (the Minkowski-sum construction has the same problem for highly obtuse angles only).

\begin{lemma}\label{lem:fw}
 Let $W$ be a facial walk in a face $F$ of a drawing of a graph $G$ in the plane. We can efficiently construct a
 nested family of inner ${\varepsilon}$-approximating polygons
 $P_{\varepsilon}$ so that
 each $P_{\varepsilon}$ has at most $\max\{3,|W|\}$ vertices.
\end{lemma}

\begin{proof}
 Let $e,v,f$ be a {\em corner} of $W$, that is, two consecutive edges $e$, $f$ and their shared vertex $v$. At $v$ erect the angle bisector of $e$ and $f$ of length $\varepsilon$ (inside $F$), and let $v'$ be the endpoint of the bisector different from $v$. For computational reasons, it may be better to use the $\ell_1$-norm at this point (the Euclidean norm will lead to square root expressions in the coordinates). If $(v_i)_{i = 1}^k$ is the sequence of vertices along $W$, with $k = |W|$, then $(v'_i)_{i = 1}^k$ defines a
 closed polygonal chain.  
 If $\varepsilon$ is sufficiently small, namely less than half the distance between any vertex of $W$ and a non-adjacent edge on $W$, the
 polygonal chain
 is free of self-crossings, and therefore bounds a simple polygon with $|W|$ vertices. There are two special cases in which this argument does not work: if the facial walk is a facial walk on an isolated vertex or an isolated edge. In both of these cases, we can approximate $W$ using a
 triangle.
\end{proof}

To prove Lemma~\ref{cor:fw} we can use
Lemma~\ref{lem:fw} to efficiently construct an inner $\varepsilon$-approximating polygon for
each facial walk of $F$.
The resulting polygons are disjoint and form a polygonal region
as long as $\varepsilon$ is less than half the distance between any two non-adjacent vertices or edges of $\cal H$.

\remove{
Lemma~\ref{lem:fw} allows us to replace a facial boundary with a
{\em polygonal region},
that is, a collection of
simple polygons 
that bound a face which is very close to the original boundary, has bounded complexity, and can be constructed efficiently. This leads to a proof of Lemma~\ref{cor:fw}. Namely, approximate each facial walk of the facial boundary with an $\varepsilon$-close
polygon 
lying in $F$. The
result is a polygonal region
as long as $\varepsilon$ is less than half the distance between any two non-adjacent vertices or edges. The upper bound of $3k$ will generally be a large overestimate, but allows for the possibility that all the inner walks are walks on isolated vertices. If there are no isolated vertices, then a walk of size $k$ gets replaced by a polygon of size at most $3k/2$ (a tight bound for $k=2$), proving the slightly sharper upper bound.
}

\subsection{Extending Partial Embeddings}\label{sec:EPP}

Our main technical tool in the proof of Theorem~\ref{thm:PWG} is the following lemma. We suggest skipping the proof of this lemma in a first reading. Multigraphs, in this paper, may have multiple edges and loops.

\begin{lemma}\label{cor:Te}
Let $G$ be a multigraph with a given planar embedding and fixed locations for a subset $U$ of its vertices. Suppose we are given a straight-line drawing of a tree $T$ whose leaves include all the vertices in $U$ at their fixed locations. Then for every $\varepsilon > 0$ there is a planar poly-line drawing of $G$ that is $\varepsilon$-close to $T$, that realizes the given embedding, where the vertices in $U$ are at their fixed locations, where each edge has at most $12n_T$ bends, and where each edge comes close to each vertex $u$ in $U$ at most six times (where coming close to $u$ means entering and leaving an $\varepsilon$-neighborhood of $u$ or terminating at $u$).
\end{lemma}

Our proof of Lemma~\ref{cor:Te} will follow closely the structure of Pach and Wenger's algorithm~\cite{PW01} to draw a planar graph with fixed vertex locations. That algorithm has three ingredients: $(i)$ making $G$ Hamiltonian, $(ii)$ drawing the Hamiltonian cycle of $G$, and $(iii)$ drawing the remaining edges of $G$. We use their result $(i)$ directly:

\begin{lemma}[{Pach, Wenger~\cite{PW01}}]\label{lem:PWHam}
  Given a planar graph $G$ we can in linear time construct a Hamiltonian graph $G'$ with $|E(G')| \leq 5|E(G)|-10$ by adding and subdividing edges of $G$ (each edge is subdivided by at most two new vertices). \end{lemma}

We will use a slightly stronger version of Lemma~\ref{lem:PWHam} in which $G$ is allowed to be a multigraph. Pach and Wenger's proof of Lemma~\ref{lem:PWHam} works in the presence of multiple edges and loops.

For part $(ii)$ Pach and Wenger show that a Hamiltonian cycle can be drawn at fixed vertex locations $\varepsilon$-close to a star connecting all the vertices. For our application, we replace their star with a straight-line drawing of a tree $T$ whose leaves are the vertices $v_i$. Lemma~\ref{lem:HCemb} shows how to draw the Hamiltonian cycle. Later we will see how to draw the remaining edges.

Independently of our result, the generalization of part $(ii)$ to trees has essentially been shown by Chan et al.~\cite{CHKL13}. Since their goal was to minimize edge lengths, they did not give an estimate on the number of bends.

\begin{lemma}
\label{lem:HCemb}
 Let $C$ be a cycle with fixed vertex locations, and suppose we are given a straight-line planar drawing of a tree $T$, in which the vertices of $C$ are leaves of $T$ at their fixed locations. Then for every $\varepsilon > 0$ there is a planar poly-line drawing of $C$ with at most $2|E(T)|-1$ bends per edge and $\varepsilon$-close to $T$.
\end{lemma}

\begin{proof}
 Let $p_1, \ldots, p_n$ be the vertices of $C$ in their order along the cycle. We build a planar poly-line drawing of $C$ as follows. Let $\Theta_i$ be an $i \varepsilon/n$-approximation of $T$ for $1 \leq i < n$ (which we construct using Lemma~\ref{lem:fw}). We start at $p_1$. Suppose we have already built the poly-line drawing of $p_1, \ldots, p_i$ and we want to add $p_ip_{i+1}$. Let $Q_i$ be the unique path in $T$ connecting  $p_i$ to $p_{i+1}$. Create $\Theta'_i$ from $\Theta_i$ by keeping only the vertices of $\Theta_i$ close to (approximating)
 vertices in $T_i := \bigcup_{j\leq i} Q_j$.
 This removes parts of the walk along $\Theta_i$ which we patch up as follows: suppose $v$ is an interior vertex of $T_i$, and
 $v$ is incident to $e$ which does not lie on $T_i$. Then $v$ is approximated by two vertices $v_1$ and $v_2$ which lie on bisectors formed by $e$ with neighboring edges. Now $v_1$ and $v_2$ belong to $\Theta'_i$, but the path along $\Theta_i$ between them got removed (since $e$ does not belong to $T_i$). We
 add $v_1v_2$ to $\Theta'_i$ to connect them. Note that $v_1v_2$ does not pass through $v$ since $v$ is incident to at least three edges ($e$ and two edges of $T_i$), and it does not cross any edges of any $\Theta'_j$ with $j<i$, since $T_i$ is monotone: if $e \not \in E(\Theta_i)$, then $e \not\in E(\Theta_j)$ for $j<i$. See Figure~\ref{fig:treecycle} for an illustration.
 \begin{figure}[tb]
\centering
\includegraphics[width=2.4in]{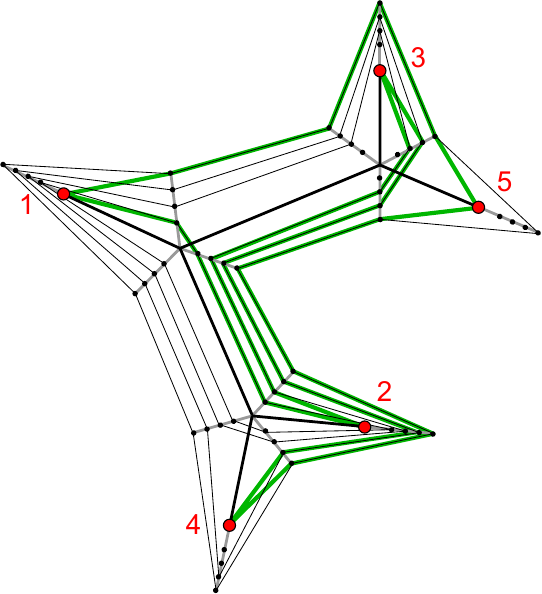}
\caption{The underlying tree $T$ is in black (thick edges), angle bisectors in gray; the $\Theta'_i$ are drawn as thin black edges; to reduce clutter, we are not showing the remaining edges of $\Theta_i$; the drawing of $C$ is indicated by the green line.}
\label{fig:treecycle}
\end{figure}
 Now both $p_i$ and $p_{i+1}$ correspond to unique vertices on $\Theta'_i$ (since they are leaves), so we can pick the facial walk $v_1, \ldots, v_k$ on $\Theta'_i$ which connects $p_i$ to $p_{i+1}$ and which avoids passing by $p_1$. We now add line segments
 $p_iv_2$, $v_2v_3$, $\ldots$, $v_{k-2}v_{k-1}$, $v_{k-1}p_{i+1}$ to the
 poly-line drawing of $C$. We treat the final edge $p_np_1$ similarly, except that we move along $\Theta'_{n-1}$ back to $p_1$ in the last step,
 which we can do, since none of the intermediate paths passed by $p_1$. Each edge of $C$ is replaced by a polygonal arc with at most $2|E(T)|-1$ bends.
\end{proof}

The following lemma shows how to draw the remaining edges of $G$, assuming that $G$ is Hamiltonian. As mentioned earlier, this lemma is close to a result by Chan et al.~\cite{CHKL13}, except for the claim about the number of bends, and the rotation system (which we need for our main result).

\begin{lemma}\label{lem:Te}
Let $G$ be a Hamiltonian multigraph with a given planar embedding and fixed vertex locations. Suppose we are given a straight-line drawing of a tree $T$ whose leaves include all the vertices of $G$ at their fixed locations. Then for every $\varepsilon > 0$ there is a planar poly-line drawing of $G$ that is $\varepsilon$-close to $T$, that realizes the given embedding, and so that the vertices of $G$ are at their fixed locations, every edge has at most $4|E(T)|-1$ bends, and every edge comes close to any leaf of $T$ at most twice.
\end{lemma}

The obvious idea---routing edges along the Hamiltonian cycle $C$---only gives a quadratic bound on the number of bends, since each edge would follow the path of a linear number of edges of $C$, and each edge of $C$ has a linear number of bends. Pach and Wenger came up with an ingenious way to construct auxiliary curves with few bends based on the level curves $\Theta'_i$ which carry the cycle $C$ in the proof of Lemma~\ref{lem:HCemb}.

\begin{proof}
 Let $C$ be the Hamiltonian cycle of $G$ and let $G_1$ and $G_2$ be the two outerplanar graphs composed of $C$ and, respectively, of the edges of $G$ outside and inside $C$.  Using Lemma~\ref{lem:HCemb} we find a planar poly-line drawing of $C$ on $V(G)$. We need to show how to draw $G_1$ and $G_2$ respecting the planar embeddings induced by the given embedding of $G$. Let $n = |V(G)|$ and $m_i = |E(G_i)|$.
 We only describe how to draw $G_1$, since $G_2$ can be handled analogously. Let $\Delta_{i,k}$, $1 \leq k \leq m_1$ be a $k\varepsilon/(nm_1)$-approximation of $\Theta'_i$ constructed using Lemma~\ref{lem:fw}. For a fixed $i$, each $\Delta_{i,k}$ crosses $C$ twice: when $C$ moves from $p_i$ to $\Theta'_{i+1}$, and when it finally moves back from $\Theta'_n$ to $p_1$. As in Pach and Wenger, we can then split $\Delta_{i,k}$ at the crossings and connect their free ends to $p_1$ and $p_i$, resulting (for each $k$) in two curves $\Delta'_{i,k}$ and $\Delta''_{i,k}$ connecting $p_1$ to $p_i$, where $\Delta'_{i,k}$ lies outside $C$ (these are the curves we use for $G_1$) and $\Delta''_{i,k}$ inside $C$ (these are the curves we use for $G_2$). Each such curve has at most $2|E(T)|-1$ bends. As in the proof of Pach and Wenger, we can create edges $p_ip_j \in E(G_1)$ by concatenating $\Delta'_{i,k}$ with $\Delta'_{j,k}$. Since we chose $m_1$ such approximations, we can do this for each edge in $G_1$. There are two problems remaining: edges $p_ip_j$ now all pass through $p_1$ and they could potentially cross (rather than just touch) there. Pach and Wenger show that any two edges touch, so the drawing can be modified close to $p_1$ so as to separate all edges $p_ip_j$ from each other. This introduces at most one more bend per edge, so that the resulting edges have $2(2|E(T)|-1) + 1 = 4|E(T)|-1$ bends. Finally, note that each edge $p_ip_j$ comes close to each leaf of $T$ (including $p_1$) at most twice, once for $\Delta'_{i,k}$ and once for  $\Delta'_{j,k}$.
\end{proof}

We are finally ready to complete the proof of Lemma \ref{cor:Te}. We show how to apply Lemma~\ref{lem:Te} in case $G$ is not Hamiltonian, and not all its vertices are assigned fixed locations.

\begin{proofof}{of Lemma~\ref{cor:Te}}
By Lemma~\ref{lem:PWHam}, we can construct a graph $G'$ with a Hamiltonian cycle $C$ by subdividing each edge of $G$ at most twice, and by adding some edges, where $G'$ has a planar embedding extending the embedding of $G$.

Next we deal with the issue that not all vertices lie in $U$, the set of vertices with fixed locations.  Traverse $C$: whenever we encounter an edge of $C$ with at least one endpoint not in $U$, contract that edge. This yields a new Hamiltonian graph $G''$ with $V(G'') = U$ and a planar embedding induced by the planar embedding of $G'$. Use Lemma~\ref{lem:Te} to construct a planar poly-line drawing of  $G''$ at the fixed vertex locations, and $\varepsilon$-close to $T$, so that each edge of $G''$ has at most $4|E(T)|-1$ bends. Each vertex $u \in U$ of $G''$ corresponds to a set of vertices $V_u \subseteq V(G')$ which was contracted to $u$, so the subgraph $G'_u$ of $G'$ induced by $V_u$ is connected. Since we embedded $G''$ with the induced planar embedding of $G'$, we can now do some surgery to turn $u$ back into $G'_u$.

The idea is to remove a small disc around vertex $u$ in the drawing of $G''$, and to draw $G'_u$ inside this disc, connected to the appropriate edges leaving the disc.
This will involve introducing new vertices where edges cross into the disc.
The same idea was used in~\cite[Theorem 2]{HJL}.

To this end, we define a graph $G^+_u$, which consists of $G'_u$, a cycle $C_u$ containing $G'_u$ in its interior, and some further edges. Each vertex of $C_u$ corresponds to an edge of $G'$ ``incident to'' $G'_u$, i.e., with an end-vertex in $V_u$ and an end-vertex not in $V_u$.
Vertices appear in $C_u$ in the same order as the corresponding edges incident to $G'_u$ leave $G'_u$ (this order also corresponds to the cyclic order of the edges incident to $u$ in $G''$); each vertex of $C_u$ corresponding to an edge $e$ of $G'$ is connected to the end-vertex of $e$ in $V_u$. Finally, $G^+_u$ contains further edges that triangulate its internal faces.

Consider a small disk $\delta$ around $u$. We erase the part of the drawing of $G''$ inside $\delta$. We construct a straight-line convex drawing of $G^+_u$ in which each vertex of $C_u$ is mapped to the point in which the corresponding edge crosses the boundary of $\delta$. This drawing always exists (and can be constructed efficiently), since $G^+_u$ is $2$-connected and internally-triangulated. Removing the edges that triangulate the internal faces of $G^+_u$ completes the reintroduction of $G'_u$.

Overall, we added one bend to an edge with exactly one endpoint in $V_u$. Since an edge can have endpoints in at most two $V_u$, this process adds at most two bends per edge, so every edge has at most $4|E(T)|+1$ bends. Since each edge of $G$ was subdivided at most twice to obtain $G'$, each edge of $G$ has at most $3 (4|E(T)|+1) = 12 |E(T)| + 3 < 12 |V(T)|$ bends. Each edge of $G'$ comes close to each leaf of $T$ at most twice, so each edge of $G$ comes close to each vertex of $U$ at most six times. This concludes the proof of Lemma~\ref{cor:Te}.
\end{proofof}

\subsection{Proof of Theorem~\ref{thm:PWG}}\label{sec:PPWG}

As we mentioned earlier, it is sufficient to prove the result for each face of ${\cal H}$, so fix such a face $F$.
Let $W_i$, with $1 \leq i \leq b$, be the facial walks of $F$. We distinguish between facial walks consisting of isolated vertices, indexed by $I := \{i: |W_i| = 1\}$, and facial walks consisting of more than one vertex, with indices in $N := \{1, \ldots, b\} - I$.
Construct an inner $\varepsilon$-approximation $P_{\varepsilon}$ of $\bigcup_{i \in N} W_i$---that is, $F$ without the isolated vertices---using Lemma~\ref{cor:fw}, and let $F'$ be the face bounded by $P_{\varepsilon}$ and isolated vertices $W_i$, $i \in I$.
For $i \in N$ let $F'_i$ be the polygon in $F'$ that approximates $W_i$.
Then
$|F_i'| \le \max\{3,|Wi|\} \le |W_i| + 1$ by Lemma~\ref{lem:fw} and the fact that $W_i$ has size at least 2.
Thus we have that $|F'| \leq \sum_{i\in N} |W_i| + |N| + |I|$.


We can triangulate $F'$ using at most $|F'| + 2|N| + |I| -4$ triangles, applying the following lemma with $n=|F'|$, $h_1=|I|$, and $h_2=|N|-1$.

\begin{lemma}[Based on O'Rourke~\protect{\cite[Lemma 5.2]{OR87}}]\label{le:number-of-triangles}
Given an $n$-vertex polygonal region with $h_1$ point-holes and $h_2$ non-point-holes, this region can be triangulated by adding chords in time $O(n \log n)$. The resulting triangulation has $n+h_1 + 2h_2 -2$ triangles.
\end{lemma}

\begin{proof}
The time bound can be derived from the algorithm of O'Rourke~\cite[Lemma 5.1]{OR87}. Consider the total sum of all angles in triangles of the triangulation. Suppose there are $n_0$ vertices on the outer face, $n_1 = h_1$ isolated vertices, and $n_2$ vertices on non-point-holes (of which there are $h_2$). Then the total angle sum is $[(n_0-2) + 2n_1 + (n_2 + 2h_2)]\pi$ which equals $t\pi$, where $t$ is the number of triangles. We conclude that $t =  n + h_1 + 2h_2-2$.
\end{proof}



We use a result of Bern and Gilbert~\cite{BG92} to construct a straight-line drawing of the dual of the triangulation.  Bern and Gilbert place a vertex at the {\em incenter} of each
triangle (where the angle bisectors of the triangle meet) and prove that the straight-line edge joining two vertices in adjacent triangles lies within the union of the two triangles. Now take a spanning tree $T$ of the dual. By Lemma~\ref{le:number-of-triangles}, $T$ has $|F'| + 2|N| + |I| -4$ vertices. For each facial walk $W_i$, $i \in N$, we augment $T$ with a new leaf  $v_i$ close to $W_i$ and inside $F'$; for each facial walk $W_i$, $i \in I$, we add the isolated vertex of $W_i$ to $T$ as a new leaf $v_i$. This adds $|N| + |I|$ vertices to $T$, so the number of vertices of $T$ is now $n_T = |F'| + 3|N| + 2|I| -4.$

Let $G_F$ be the embedded multigraph obtained by restricting $G$ to vertices and edges lying inside or on the boundary of $F$ and by contracting each facial walk $W_i$ of $F$ to a single vertex $v_i$. We can now use Lemma~\ref{cor:Te} to embed $G_F$ along $T$ so that vertices $v_i$ are drawn at their fixed locations. Each edge of $G_F$ has at most $12n_T$ bends.

\begin{figure}[tb]
\centering
\includegraphics[width=0.8\textwidth]{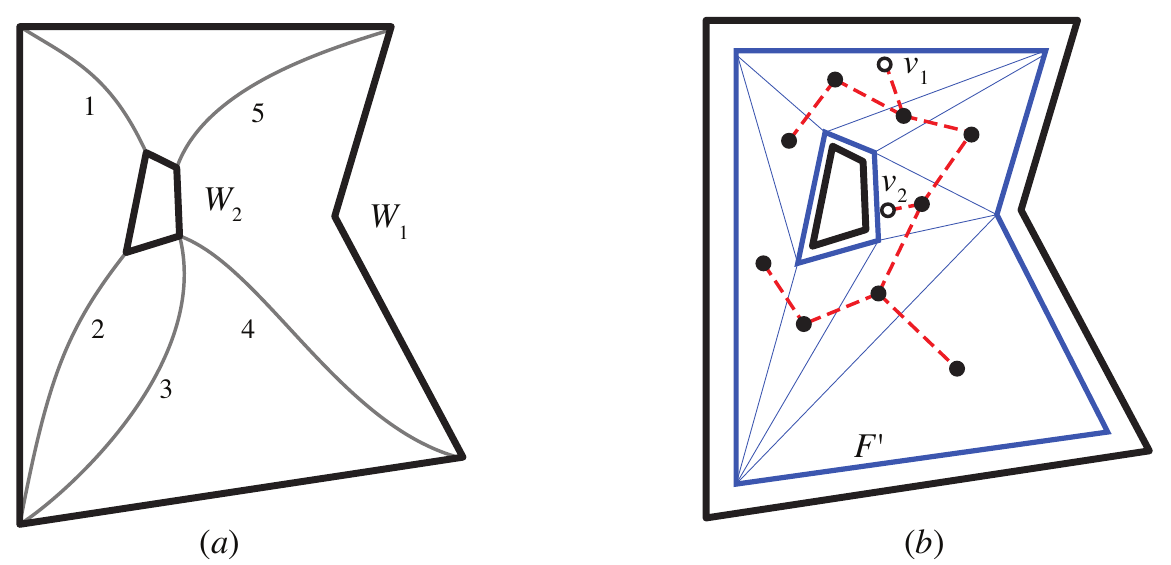}
\caption{A face $F$ with outer facial walk $W_1$ and inner facial walk $W_2$.  (\emph{a}) The 5 edges of $G-H$.  (\emph{b}) The inner approximation $F'$ (heavy blue lines), a triangulation of it (fine lines), and the dual spanning tree (dashed red) with extra vertices $v_1$ and $v_2$ close to  $W_1$ and $W_2$, respectively.}

\label{fig:route1}
\end{figure}

We now want to connect edges in $G_F$ to the boundary components they belong to. For facial walks $W_i$, $i \in I$, there is nothing to do, since we chose $v_i$ to be the isolated vertex which is the boundary component $W_i$. So we may assume that we are dealing with boundary components consisting of more than one vertex. We will use the buffer between $F'$ and $F$ to do this. In fact, we need to split the buffer zone into two,  so we apply
Lemma~\ref{cor:fw} a second time to obtain an inner $\varepsilon/2$-approximation $F''$ of $F$, so that $F' \subseteq F'' \subseteq F$. See Figure~\ref{fig:route23}. Let $F''_i$ be the polygon that approximates $W_i$ in $F''$. Note that $|F_i''| = |F_i'| \le |W_i| + 1$.
Now for each walk $W_i$ we extend the edges ending at $v_i$ to their endpoint on $W_i$. Since we maintained the cyclic order of $G_F$-edges at $v_i$, we can simply route these edges around $W_i$ using approximations to $W_i$ via Lemma~\ref{cor:fw}, and we can do so in $F_i-F_i''$.
This adds two bends to the edge near $v_i$, plus at most one bend for each vertex of $F_i''$ except the one corresponding to the final destination vertex on $W_i$.  In total we add at most $2 + |F_i''|-1 \le |W_i| + 2$ bends.
There is one difficulty: there are edges of $G_F$ that pass by $v_i$, separating it from the segment of $F'$ close to $v_i$ (which is our gate to $W_i$). To remedy this difficulty, we first route all of these edges around the whole obstacle $W_i$ in the $F''-F'$ part of the buffer, which adds $|F_i'| + 3 \le |W_i| + 4$ bends to an edge every time it passes $v_i$ (see Figure~\ref{fig:route23}$(b)$, note that the edge starts with one bend close to the vertex).


Now we are free to route the $G_F$-edges incident to $v_i$ to their endpoints along $W_i$. Since an edge can pass by and/or terminate at a vertex at most six times, the number of additional bends in each edge caused by going around
$W_i$ is at most $6(|W_i| + 4) = 6|W_i| + 24$; totalling this number over all boundary components of $F$ yields a bound of at most
$6 \sum_{i\in N} |W_i| + 24|N|$ bends along the whole edge (we can ignore $W_i$ with $i \in I$, since we do not reroute around those components). Since each $G_F$-edge started with $12n_T$ bends, each $G_F$-edge now has at most
$12n_T + 6\sum_{i\in N} |W_i| + 24|N|$ bends.


In order to derive a bound in terms of $n_H=|V(H)|$, we use:

$(1)$ $n_T = |F'| + 3|N| + 2|I| -4$  (as discussed in the first part of this subsection),

$(2)$ $|F'| \leq \sum_{i\in N} |W_i| + |N| + |I|$ (as discussed in the first part of this subsection),


$(3)$ $\sum_{i\in N} |W_i| \leq 2n_H$ (which can be easily proved by induction on $|N|$, primarily, and on the number of $2$-connected components of $W_i$, if $|N|=1$), and
 
 
 $(4)$ $2|N| + |I| \leq n_H$ (since each facial walk $W_i$ with $i\in N$ consists of more than one vertex).

\remove{In order to derive a bound in terms of $n_H=|V(H)|$, we use $(1)$ $n_T = |F'| + 3b-4$ (as discussed in the first part of this subsection), $(2)$ $|F|=\sum_{i\in N} |W_i| + |I| \leq 2n_H$, $(3)$ $b = |N| + |I|$, $(4)$ $2|N| + |I| \leq n_H$ (since each facial walk $W_i$ with $i\in N$ consists of more than one vertex), and $(5)$ $|F'| \leq 3/2 \sum_{i\in N} |W_i| + |I|$ (as discussed in the first part of this subsection).
}

From (1) and (2) we get that $n_T \leq \sum_{i\in N} |W_i| + 4|N| + 3|I|$.
Thus the number of bends in each $G_F$-edge is at most
\begin{align*}
12n_T + 6\sum_{i\in N} |W_i| + 24|N|
&\leq 12( \sum_{i\in N} |W_i| + 4|N| + 3|I|) + 6\sum_{i\in N} |W_i| + 24|N| \\
&\leq 18 \sum_{i\in N} |W_i| + 72 |N| + 36 |I| \\
&\leq 18 (\sum_{i\in N} |W_i|) + 36 (2|N| + |I|).
\end{align*}
From (3) and (4), we conclude that each $G_F$-edge has at most $36n_H + 36 n_H = 72 n_H$ bends.




\begin{figure}[tb]\centering
\includegraphics[width=\textwidth]{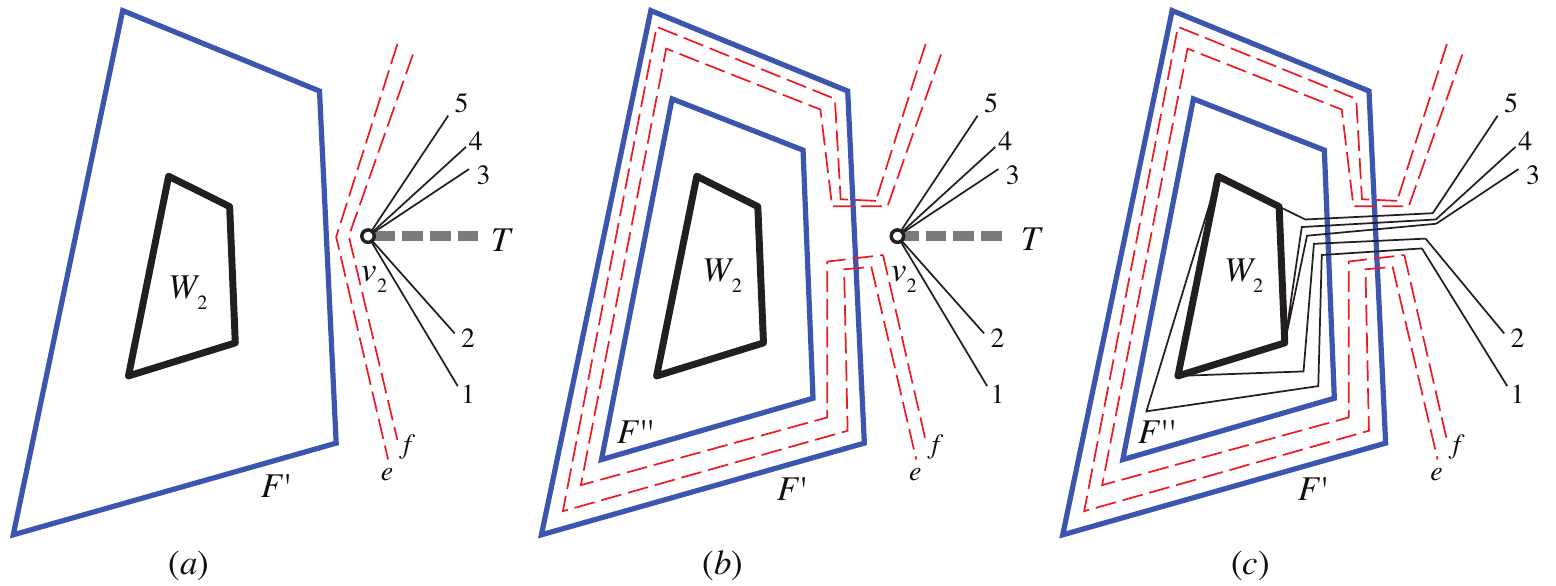}
\caption{A close-up of the situation near inner facial walk $W_2$.
(\emph{a}) After drawing $G_F$ around the tree $T$ (heavy dashed line), edges
$1, \ldots, 5$ are incident to $v_2$ in the correct cyclic order, but two other edges $e$ and $f$  pass by between $v_2$ and $F'$.
(\emph{b})  We add a second approximation $F''$ and route the edges $e$ and $f$ (in dashed red) around $W_2$  in the buffer zone between $F''$ and $F'$.
(\emph{c})  We route the edges incident to $W_2$ in the buffer zone between $F$ and $F''$.}
\label{fig:route23}
\end{figure}

Most of the steps in the construction can be performed in linear time. Building the triangulation takes time $O(n_H \log n_H)$. The overall running time is thus bounded by the size of the resulting drawing which contains a linear number of edges each with a linear number of bends, yielding the quadratic running time.

{\bf Remark 1.} The algorithm we presented in this section provides a bound better than $72 n_H$ bends per edge if the subgraph $H$ of $G$ for which a straight-line drawing $\cal H$ is given as part of the input is {\em induced}. If that is the case, then the embedded multigraph $G_F$ defined in this section contains no self-loops; consequently, a Hamiltonian graph $G'_F$ can be constructed in linear time by adding vertices and edges and by subdividing edges of $G_F$ so that each edge is subdivided by at most one new vertex (while in the general case we use two subdivision vertices per edge, see Lemma~\ref{lem:PWHam}). This immediately allows us to improve the bounds in Lemma~\ref{cor:Te} on the number of bends per edge to $8 n_T$ and on the number of times each edge comes close to each vertex $u$ to at most four. The same analysis as above and the improved bounds of Lemma~\ref{cor:Te} allow us to upper bound the number of bends per edge in Theorem~\ref{thm:PWG} by $48 n_H$.

{\bf Remark 2.} An improvement upon the $72 n_H$ bound of Theorem~\ref{thm:PWG} can be obtained by modifying the placement of $v_i$, for each $i \in N$, and the route of the edges that go around $W_i$. This modification makes the algorithm slightly more involved, so we preferred to omit it from the proof and to sketch it here. The main idea is that vertex $v_i$ can be inserted not just at any point inside $F'$, but rather at a convex corner of $F'_i$ that approximates an occurrence $\sigma$ of a vertex of $W_i$. Then each edge that goes around $v_i$ and has to be ``wrapped around'' $W_i$ can save three bends (each time it passes by $v_i$) with respect to the route described in Figure~\ref{fig:route23}$(b)$. To achieve this, we bend the edge at its intersection points with $F'_i$ and then connect it directly to the suitable approximations of the vertices next to $\sigma$ along $W_i$. This route introduces $|F'_i|=|W_i| + 1$ new bends each time an edge passes by $v_i$. A similar argument can be used for the edges that terminate at some vertex of $W_i$. This results in each $G_F$-edge having at most $12n_T+6\sum_{i\in N} |W_i| + 6|N|$ bends. Then the same calculations described above lead to a bound of $63n_H$ bends per edge.

\section{Extending Partial Drawings Greedily} \label{se:greedy}

Let $G$ be a plane graph with a spanning subgraph $H$ for which we have fixed a straight-line planar drawing ${\cal H}$.
For a given ordering $\sigma=[e_1,\dots,e_m]$ of the edges in $G\setminus H$ we say that a drawing $\Gamma$ of $G$ {\em greedily extends $\cal H$ with respect to $\sigma$} if it is obtained by drawing edges $e_1,\dots,e_m$ in this order, so that $e_i$ is drawn as a polygonal curve that respects the embedding of $G$ and with the minimum number of bends, for $i=1,\dots,m$.

Suppose $\sigma$ orders the edges of $G \setminus H$ so that the edges between distinct connected components of $H$ precede edges between vertices in the same connected component of $H$. For such orderings Fowler {\em et al.}\ claimed in~\cite{fjks-crp-11} that there exists a drawing $\Gamma$ of $G$ greedily extending $\cal H$ with respect to $\sigma$ in which each edge has $O(|V(G)|)$ bends. However, in the following we confirm a claim of Schaefer~\cite{S13a} stating that greedy extensions do not, in general, lead to drawings with a polynomial number of bends.

\begin{theorem}\label{th:greedy}
For every $n$ there exists an $n$-vertex plane graph $G$, a planar drawing $\cal H$ of $H = (V(G), \emptyset)$, the empty spanning subgraph of $G$, and an order $\sigma$ of the edges in $G$ so that any drawing of $G$ that greedily extends $\cal H$ with respect to $\sigma$ has edges with $2^{\Omega(n)}$ bends.
\end{theorem}

\begin{proof}
We adapt an example by Kratochv\' il and Matou\v sek~\cite{km-sgrer-91}. Refer to Figure~\ref{fig:exponential}. Let $N=\lfloor \frac{n}{3} \rfloor -6$, for any integer $n$. Graph $H$ consists of $n$ isolated vertices, name them $u_1,\dots,u_N, v_1,\dots,v_N,w_1,\dots,w_N,a,b,c,d,e,r_1,\dots,r_{n-3N-5}$. The first $n-N-1$ edges in $\sigma$ are $(u_i,w_i)$ for $i=1,\dots,N$, $(w_i,w_{i+1})$ for $i=1,\dots,N-1$, $(r_i,r_{i+1})$ for $i=1,\dots,n-3N-6$, $(c,w_1)$, $(b,c)$, $(c,e)$, $(e,d)$, $(a,d)$, and $(a,r_{n-3N-5})$. All these edges are straight-line segments in any drawing $\Gamma$ of $G$ that greedily extends $\cal H$ with respect to $\sigma$. The last $N$ edges in $\sigma$ are $(u_1,v_1),\dots,(u_N,v_N)$ in this order.

\begin{figure}[tb]
\centering
\includegraphics[scale=0.9]{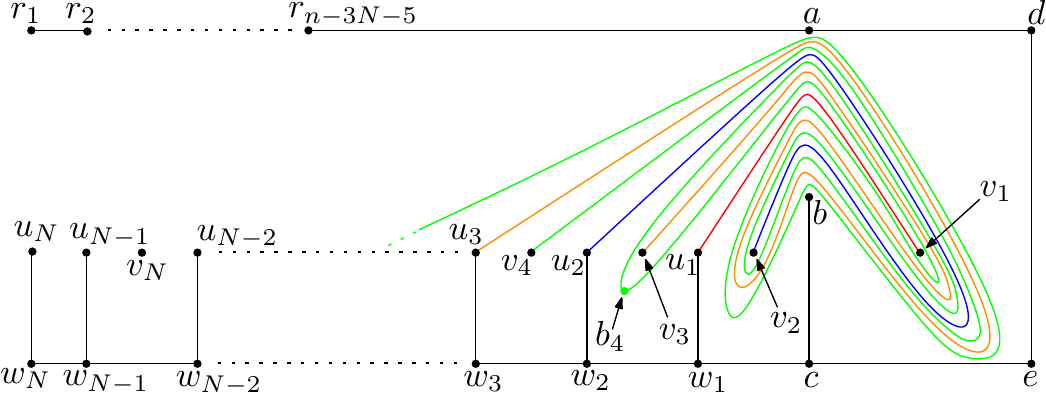}
\caption{A drawing $\Gamma$ of $G$ that greedily extends $\cal H$ with respect to $\sigma$. Drawing $\cal H$ consists of the black circles. The first edges $n-N-1$ edges in $\sigma$ are (black) straight-line segments. The last $N$ edges $(u_i,v_i)$ are (colored) polygonal lines whose bends have been made smooth to improve the readability. Only four of the latter edges are shown.}
\label{fig:exponential}
\end{figure}

Consider any drawing $\Gamma$ of $G$ that greedily extends $\cal H$ with respect to $\sigma$. We claim that edge $(u_i,v_i)$ has $2^{i-1}$ bends in $\Gamma$. In fact, it suffices to prove that $(u_i,v_i)$ has $2^{i-1}$ intersections with the straight-line segment $ab$ in $\Gamma$. Indeed, $(u_1,v_1)$ has exactly one intersection with $ab$ in $\Gamma$. Inductively assume that $(u_{i},v_{i})$ has $2^{i-1}$ intersections with $ab$ in $\Gamma$; we prove that $(u_{i+1},v_{i+1})$ has $2^{i}$ intersections with $ab$ in $\Gamma$. This proof is accomplished by following Kratochv\' il and Matou\v sek~\cite{km-sgrer-91} almost {\em verbatim}. Since $(u_{i+1},v_{i+1})$ does not cross $(u_{i},v_{i})$, it has a bend $b_{i+1}$ around $v_i$, i.e., inside the square defined by $u_{i-2}$, $w_{i-2}$, $w_{i-1}$, and $u_{i-1}$. Thus the polygonal curve representing $(u_{i+1},v_{i+1})$ in $\Gamma$ consists of two parts---one from $u_{i+1}$ to $b_{i+1}$, the other from $b_{i+1}$ to $v_{i+1}$. Both of these parts may be used as an edge joining $u_i$ and $v_i$, after contracting $u_{i+1}$ and $v_{i+1}$ into $u_i$, and $b_{i+1}$ into $v_i$. Hence, by induction, each of these two parts has $2^{i-1}$ intersections with $ab$, and the whole edge $(u_{i+1},v_{i+1})$ has $2^{i}$ intersections with~$ab$.

Hence, in any drawing $\Gamma$ of $G$ that greedily extends $\cal H$ with respect to $\sigma$, one edge has $2^{N-1}=2^{\lfloor \frac{n}{3} \rfloor -7}\in 2^{\Omega(n)}$ bends, which concludes the proof.
\end{proof}

We remark that the graph $G$ in the proof of Theorem~\ref{th:greedy} is a tree, so every edge of $G$ connects vertices in distinct connected components of $H$.

\section{Simultaneous Planarity}
\label{sec:S2}


Before turning to our algorithm to draw simultaneously planar graphs, we justify our claim that minimizing the number of crossings in a simultaneous planar drawing is \NP-hard. This result follows from Cabello and Mohar's proof of  \NP-hardness for the \emph{anchored planarity} problem~\cite[Theorem 2.1]{Cabello-Mohar}, but a more direct proof of a slightly stronger result is possible by reduction from the \NP-complete crossing number problem.

\begin{theorem}\label{thm:SPC}
 Minimizing the number of crossings in a simultaneous planar drawing of two graphs is \NP-complete, even if one graph is the disjoint union of paths of length at most two and the other graph is a matching.
\end{theorem}

The result is sharp in the sense that if both $G_1$ and $G_2$ are matchings, the problem is easy, since the union of two matchings is always planar.

\begin{proof}
We use the fact that the (standard) crossing number problem is \NP-hard for cubic graphs~\cite{H06}. Let $G$ be a cubic graph with $m$ edges.
Subdivide each edge $2m$ or $2m+1$ times (we will shortly see which). At each of the original vertices of $G$ choose two of the incident edges, and make them part of $G_1$; the third edge at each vertex is added to $G_2$. Now add the remaining edges to $G_1$ and $G_2$ so that along each path between original vertices $G_1$ and $G_2$ edges alternate. If such a path ends with two $G_1$-edges or two $G_2$-edges, we need to subdivide it $2m$ times to make this possible; if it ends with one $G_1$-edge and one $G_2$-edge, we subdivide it $2m+1$ times. By this construction, $G_1$ is a disjoint union of paths of length at most two, and $G_2$ is a matching. Finally, the number of crossings in a simultaneous planar drawing of $G_1$ and $G_2$ is an upper bound on the crossing number of $G$, and, since we subdivided each edge of $G$ sufficiently often, the two numbers are equal: starting with a crossing-minimal drawing of $G$, we can realize each crossing by aligning a $G_1$-edge with a $G_2$-edge.
\end{proof}


We now turn to the proof the Theorem~\ref{thm:sim-draw}.

\begin{proofof}{of Theorem~\ref{thm:sim-draw}}
We first note that it is easy to go from $(ii)$ to $(i)$: Suppose
we have constructed, in time $O(n^2)$ a simultaneous planar drawing $\Gamma$ so that private edges of $G_1$ and $G_2$ intersect at most $24$ times, all edges of $G_1$ are straight, and all private edges of $G_2$ have at most $72 |V(G_1)|$ bends. We add dummy
vertices at the locations of the $O(n^2)$ crossings points in $\Gamma$, and then construct a straight-line drawing of the resulting planar graph on a small grid. The number of bends per edge in the new drawing is at most $72n$, since each edge in $\Gamma$ intersects fewer than $3n$ edges, and each one of them at most $24$ times.

We are left with the proof of $(ii)$. That is, we have to construct in time $O(n^2)$ a simultaneous planar drawing of $G$ in which private edges of $G_1$ and $G_2$ intersect at most $24$ times, all edges of $G_1$ are straight, and every private edge of $G_2$ has at most $72|V(G_1)|$ bends.

Start with an arbitrary straight-line planar drawing $\Gamma_1$ of $G_1$. We now construct a drawing $\Gamma_2$ of $G_2$ using an approach similar to the proof of Theorem~\ref{thm:PWG}. Drawing $\Gamma_1$ induces a straight-line planar drawing $\Gamma$ of $G$. Thus, in order to determine $\Gamma_2$, it remains to describe how to draw the private edges of $G_2$. We will accomplish this independently for each face $F$ of $G$.

We construct a triangulation $\Sigma$ of $F$ by using all the vertices and edges of $G_1$ that lie inside $F$, as well as some extra edges we will specify shortly. Next, we execute the same algorithm we used in the proof of Theorem~\ref{thm:sim-draw}.  Namely, we construct a straight-line drawing of the dual $D$ of $\Sigma$ and we take a spanning tree $T$ of $D$. For each facial walk $W_i$ of $F$, we augment $T$ with a leaf  $v_i$ close to $W_i$ and inside $F'$, if $|W_i| > 1$, and coinciding with $W_i$, if $|W_i| = 1$; here, $F'$ is an inner $\varepsilon$-approximation of $F$ constructed as earlier. Let $G^F_2$ be the embedded multigraph obtained by restricting $G_2$ to the vertices and edges inside or on the boundary of $F$, and by contracting each facial walk $W_i$ of $F$ to a single vertex $v_i$. We use Lemma~\ref{cor:Te} to construct a planar poly-line drawing of $G^F_2$ that realizes the given embedding, that is $\varepsilon$-close to $T$, and in which vertices $v_i$ maintain their fixed locations. Finally, for boundary components with $|W_i| > 1$, we reconnect edges in $G^F_2$ to the boundary components they belong to. In order to do this, we first ``wrap'' the edges of $G^F_2$ passing by a vertex $v_i$ around $W_i$, and we then extend the edges of $G^F_2$ incident to $v_i$ to their endpoint on $W_i$, by routing them around $W_i$.

By construction every edge of $G_1$ is straight. By Theorem~\ref{thm:PWG} every private edge of $G_2$ has at most $72 |V(G_1))|$ bends. Also, the algorithmic steps are the same as for the proof of Theorem~\ref{thm:PWG}, hence the algorithm runs in $O(n^2)$ time. It remains to prove that any private edge of $G_1$ and any private edge of $G_2$ intersect at most $24$~times.

Consider any private edge $e$ of $G_2$ and any private edge $e'$ of $G_1$. Recall that $e'$ is an edge of $\Sigma$. Denote by $W_i$ and $W_j$ the facial walks that the end-vertices of $e'$ belong to.
Edge $e$ can only intersect edge $e'$ in the following two situations:
when passing by $v_i$ or $v_j$ and when passing by the point $p_T$ in which the edge of $D$ dual to $e'$ crosses $e'$. We prove that each of these two types of intersections happens at most $12$ times.

For the first type of intersections, Lemma~\ref{cor:Te} implies that edge $e$ passes by each of $v_i$ or $v_j$ at most $6$ times, hence at most $12$ times in total.

For the second type of intersections, Lemma~\ref{lem:PWHam} implies that edge $e$ is subdivided into at most three edges $e_1$, $e_2$, and $e_3$ in order to turn $G^F_2$ into a Hamiltonian graph. For each $j=1,2,3$, $e_j$ either belongs to the Hamiltonian cycle of the subdivided $G^F_2$ or not. In the former case, $e_j$ is drawn as part of an $i \varepsilon/n$-approximation $\Theta_i$ of $T$, as in the proof of Lemma~\ref{lem:HCemb}, hence it crosses $e'$ at most twice. In the latter case, $e_j$ is composed of two parts, denoted by $\Delta'_{p,k}$ and $\Delta'_{q,k}$, or by $\Delta''_{p,k}$ and $\Delta''_{q,k}$ in the proof of Lemma~\ref{lem:Te}. Each of $\Delta'_{p,k}$, $\Delta'_{q,k}$, $\Delta''_{p,k}$ and $\Delta''_{q,k}$ is part of a $k\varepsilon/(nm_1)$-approximation of $\Theta'_i$, which is part of $\Theta_i$. Hence, each of $\Delta'_{p,k}$, $\Delta'_{q,k}$, $\Delta''_{p,k}$ and $\Delta''_{q,k}$ crosses $e'$ at most twice; thus $e_j$ crosses $e'$ at most four times, and $e$ crosses $e'$ close to $p_T$ at most $12$ times.
\end{proofof}

\section{Open Questions}

We conclude with three open questions.
We proved that if a graph has a planar drawing extending a straight-line planar drawing of a subgraph then there is such a drawing with at most $72n$ bends per edge.  This is asymptotically tight, but can the constant $72$ be reduced? As sketched at the end of Section~\ref{sec:PEG}, a variation of our algorithm decreases this constant to $63$, however new ideas seem to be needed in order to push the bound further down.

Our second result was that any two simultaneously planar graphs have a simultaneous planar drawing with at most $24$ crossings per pair of edges and a linear number of bends per edge with a drawing on a polynomial-sized grid.  The only lower bound on the number of crossings between two edges in a simultaneous planar drawing is 2 (see~\cite{CJS08} or the figure in the margin for the entry ``simultaneous crossing number'' in~\cite{S13b}).  There is a large gap between $2$ and $24$.  Can two edges be forced to cross more than twice in a simultaneous planar drawing?
For the third open question, we note that
Grilli et al.~\cite{Grilli2014} showed that two simultaneously planar graphs have a drawing with at most $9$ bends per edge, though with a larger constant for the number of crossings and not on a grid.  Is it possible to achieve the best of both results:  $9$ bends per edge, $24$ crossings per pair of edges, and a nice grid?


\bibliographystyle{abbrvurl}
\bibliography{./SEFE2}

\end{document}